\newcommand{\mytitle}{Sparsity Preserving Algorithms for Octagons}
\documentclass{entcs} \usepackage{entcsmacro}
\nonstopmode

\usepackage[T1]{fontenc}
\usepackage[utf8]{inputenc}
\usepackage{amsmath}
\usepackage{amssymb}
\usepackage{color}
\usepackage{verbatim}
\usepackage{galois}
\usepackage{enumerate}
\usepackage{upgreek}

\newcommand{\secref}[1]{\hyperref[#1]{\S\ref*{#1}}}
\newcommand{\figref}[1]{\hyperref[#1]{Fig.~\ref*{#1}}}
\newcommand{\thmref}[1]{\hyperref[#1]{Theorem~\ref*{#1}}}
\newcommand{\defref}[1]{\hyperref[#1]{Definition~\ref*{#1}}}
\newcommand{\exref}[1]{\hyperref[#1]{Example~\ref*{#1}}}

\catcode`"=\active \def"#1"{\text{\texttt{#1}}} 

\hypersetup{pdftitle={\mytitle},pdfauthor={Jacques-Henri Jourdan}}

\def\lastname{Jourdan}
\begin{document}
\begin{frontmatter}
\title{\mytitle}
\author{Jacques-Henri Jourdan}
\address{MPI-SWS, Inria Paris}
\thanks{This work was supported by Agence Nationale de la Recherche,
  grant ANR-11-INSE-003.}
\begin{abstract}
  Known algorithms for manipulating octagons do not preserve their
  sparsity, leading typically to quadratic or cubic time and space
  complexities even if no relation among variables is known when they
  are all bounded. In this paper, we present new algorithms, which use
  and return octagons represented as weakly closed difference bound
  matrices, preserve the sparsity of their input and have better
  performance in the case their inputs are sparse. We prove that these
  algorithms are as precise as the known ones.
\end{abstract}
\end{frontmatter}

\newcommand\lsharp{{}^\sharp\!\mathord}

\section{Introduction}

In order to capture numerical properties of programs, static analyzers
use \emph{numerical abstract domains}. The choice of a numerical
abstract domain in a static analyzer is a compromise between
\emph{precision}, the ability of capturing complex numerical
properties, and performance. Non-relational abstract domains, such as
intervals~\cite{cousot1976static}, are very efficient but relatively
imprecise: they cannot represent relations between program
variables. On the other hand, in order to capture numerical relations
between program variables, one can express them as linear
inequalities. This class of relational numerical abstract domain is
composed of \emph{linear abstract domains}. A linear abstract domain
corresponds to a different precision vs.\ performance trade-off: they
range from the less precise, efficient ones such as
zones~\cite{mine2004phd}, pentagons~\cite{logozzo2008pentagons} or
octagons~\cite{mine2004phd,mine2006octagon} to the more precise,
costly ones, such as subpolyhedra~\cite{laviron2009subpolyhedra},
octahedra~\cite{clariso2004octahedron}, two variables per
inequalities~\cite{simon2010two}, zonotopes~\cite{goubault2006static}
or general polyhedra~\cite{cousot1978automatic}.

In particular, the Octagon abstract
domain~\cite{mine2004phd,mine2006octagon} accurately represents many
of the variable relationships appearing in a program, while being
still reasonably fast (all the operations have quadratic or cubic
complexity on the number of variables). It is very popular in the
static analysis community, which explains why algorithmic
improvements~\cite{chawdhary2014simple,bagnara2009weakly,singh2015making}
and precision improving variants~\cite{chen2014abstract} are regularly
published.

As reported by the designers of Astrée~\cite{cousot2009does}, its
quadratic or cubic performances still make it unusable as-is with a
reasonable number of variables. Indeed, the data structures typically
used to represent octagonal abstract values, i.e., strongly closed
difference bound matrices, have a quadratic size in the number of
variables for which an upper or lower bound is known. A common
solution is the use of \emph{variable packing}~\cite[\S
8.4.2]{mine2004phd}, where the Octagon abstract domain is only used on
small packs of variables. The downside of packing is that no relation
is stored between variables that are not in the same pack. A variant
of packing has been introduced to mitigate the
imprecision~\cite{bouaziz2012treeks}, but loss in precision can still
occur.

The problem of the performance of octagons has already been studied:
in particular, Singh et al.~\cite{singh2015making} proposed an
implementation of the Octagon abstract domain optimized in the case
its representation is sparse. But they do not address the fact that it
is dense as soon as interval bounds are known for many variables, and
we anticipate that, for this reason, the sparsity is very low in their
implementation.

Instead, in this paper, we propose to use new algorithms for the
Octagon abstract domain: these algorithms work on a sparse
representation for octagons, so that the cost of the analysis of two
independent sets of variables is the sum of the costs of the analyses
of the two sets of variables, taken independently. Our algorithms have
the same precision as the traditional ones. Our main idea is the
following: in order to ensure an optimal precision of all the
operations, the data structures representing octagons, difference
bound matrices, are usually kept \emph{strongly closed}: that is,
algorithms make sure that any returned difference bound matrix is a
best abstraction. However, most often, strongly closed difference
bound matrices are dense because of the necessary \emph{strengthening}
step. In this paper, we propose to weaken the maintained invariant on
difference bound matrices and to keep them \emph{weakly closed} hence
skipping the strengthening step. Weakly closed difference bound
matrices are not necessarily dense, so that we can use sparse data
structures to represent them. We prove that some algorithms can be
kept unchanged to work on weakly closed difference bound matrices
without losing any precision and give new algorithms for the other
operations.

We begin by preliminary definitions in \secref{sec:diffbndmat}. In
\secref{sec:dbmops}, we describe and prove the soundness and relative
precision of our new algorithms. We conclude in
\secref{sec:conclusion}.

\section{Definitions}
\label{sec:diffbndmat}

\newcommand\varplus{\mathbb V\!_+}
\newcommand\varpm{\mathbb V\!_\pm}

Let $\varplus$ be a finite set of variables. We call a \emph{regular
  environment} a function from $\varplus$ to $\mathbb R$. A regular
environment represents the numerical state of a program. The role of
the Octagon abstract domain is to approximate sets of regular
environments~$\rho$. To that end, the abstract domain of octagons
stores a set of inequalities of the following form:
\begin{equation}
\label{eq:octConstr1}
\pm\rho(u)\pm\rho(v) \leq Cst_{uv} \qquad u,v\in\varplus
\end{equation}
This corresponds to giving bounds to sums and differences of values of
$\rho$. Moreover, if we use twice the same variable with the same
sign, we see that, using such constraints, we can express interval
constraints over values of an environment~\cite{mine2004phd}.

In order to handle in a unified way all the different combinations of
signs in these constraints, we introduce the set $\varpm$ of
\emph{signed variables}. Signed variables are of two kinds: they are
either usual variables from $\varplus$, called \emph{positive
  variables} in the context of signed variables, or their opposites
form, \emph{negative variables}. We equip $\varpm$ with an involutive
operator, associating to each signed variable $v$ its opposite
$\overline v$, such that $v$ is positive if and only if $\overline v$
is negative.

Regular environments are canonically extended to signed variables by
taking $\rho(\overline u) = -\rho(u)$. More generally, we define
\emph{irregular environments} as functions $\sigma$ from $\varpm$ to
$\mathbb R$, and consider the set of regular environments as a subset
of the set of irregular environments. Regular environments are exactly
irregular environments $\rho$ that satisfy the property $\forall v,
\rho(\overline v) = -\rho(v)$.

Using this new formalism, octagonal constrains of the
form~\eqref{eq:octConstr1} can be seen as upper bounds on differences
of values of $\rho$, a regular environment:
\begin{equation}
\label{eq:octConstr2}
\rho(u)-\rho(v) \leq Cst_{uv} \qquad u,v\in\varpm
\end{equation}
This has two benefits: first, all the different kinds of constraints
allowed by~\eqref{eq:octConstr1} get factored out as one simpler
form. Second, we can see these constraints as constraints on irregular
environments, and further constrain them as being regular: we see that
the study of the Octagon abstract domain starts by the study of a
simpler abstract domain, where only differences of variables are
bounded. The set of constraints, called \emph{potential constraints},
of such an abstract domain is well studied in the linear optimization
literature, because it corresponds to the well-known shortest path
problem in a weighted directed graph.

Such a set of constraints is represented as difference bound matrices:
a \emph{difference bound matrix}, or \emph{DBM}, is a matrix
$(B_{uv})_{(u,v)\in\varpm^{\,2}}$ of elements of $\mathbb
R\cup\{+\infty\}$. The meaning of these constraints is given by two
concretization functions $\gamma_{pot}$ and $\gamma_{oct}$, that
associate to a DBM the set of irregular or regular environments,
respectively, satisfying all the constraints:
\begin{align}
  \gamma_{pot}(B) &= \{ \sigma : \varpm \rightarrow
  \mathbb R \mid \forall uv \in \varpm,\;\sigma(u)-\sigma(v)
  \leq B_{uv}\} \\
  \gamma_{oct}(B) &= \{ \rho \in \gamma_{pot}(B) \mid
  \forall u \in \varpm,\; \rho(\overline u) = - \rho(u) \}
\end{align}

\begin{example}
  \label{ex:firstex}
  Consider $\varplus = \{ x; y; z \}$ a set of three (positive)
  variables. The set of signed variables is $\varpm = \{ x ; \overline
  x; y; \overline y; z; \overline z\}$. Let $A$ be the DBM such that
  $A_{x \overline x} = 1$, $A_{\overline y y} = 3$, $A_{y \overline z}
  = 1$ and $A_{uv} = +\infty$ for all the other entries. The set
  $\gamma_{oct}(A)$ contains all the environments $\rho : \varpm
  \rightarrow \mathbb{R}$ such that:
  \begin{itemize}
  \item $\forall u \in \varpm,\; \rho(\overline u) = - \rho(u)$
  \item $\rho(x) \leq 1/2$, $-\rho(y) \leq 3/2$ and $\rho(y) + \rho(z) \leq 1$
  \end{itemize}
  This concretization is assimilated to the set of environments $\rho
  : \varplus \rightarrow \mathbb{R}$ over positive variables such that
  $\rho(x) \leq 1/2$, $-\rho(y) \leq 3/2$ and $\rho(y) + \rho(z) \leq
  1$.
\end{example}

\newcommand\dbmleq{\mathrel{\lsharp{\,\leq}}}

We denote as $\dbmleq$ the natural order relation over DBMs, defined
as follows:
\begin{equation}
  A \dbmleq B \Leftrightarrow
  \forall uv \in \varpm,\;A_{uv} \leq B_{uv}
\end{equation}
The following easy lemma states that this order relation makes
$\gamma_{pot}$ and $\gamma_{oct}$ increasing, which makes $\dbmleq$ a
good candidate for a comparison operator of the Octagon abstract
domain\footnote{As we see in \secref{sec:octwcomp}, this is \emph{not}
  the order relation we use as a comparison operator in our
  implementation of the Octagon abstract domain.}:
\begin{lemma}
\label{lem:octgammamonot}
Let $A$ and $B$ be two DBMs such that $A \dbmleq B$. Then, we
have:
\begin{align*}
  \gamma_{pot}(A) &\subseteq \gamma_{pot}(B) &
  \gamma_{oct}(A) &\subseteq \gamma_{oct}(B)
\end{align*}
\end{lemma}

For any non-empty set $S$ of irregular environments, there exists a
minimal (in the sense of $\dbmleq$) DBM that approximates
it. That is, there exists a minimal DBM $\alpha(S)$ such that $S
\subseteq \gamma_{pot}(\alpha(S))$. This property follows immediately
from the definition of~$\alpha$:
\begin{equation}
\alpha(S)_{uv} = \sup_{\sigma\in S} \{ \sigma(u) - \sigma(v)\}
\end{equation}
This function $\alpha$ is called the abstraction function. We can
easily see that $\alpha$ is an increasing function. Moreover, $\alpha$
does not only return best abstractions for $\gamma_{pot}$, but also
for $\gamma_{oct}$: if the set $S$ contains only regular environments,
we can see that $\alpha(S)$ is also the minimal DBM such that $S
\subseteq \gamma_{oct}(\alpha(S))$.  In fact, it is easy to see that
$\dbmleq$ defines a complete lattice over DBMs extended with a bottom
element, and that the pairs $(\alpha, \gamma_{pot})$ and
$(\alpha, \gamma_{oct})$ form Galois connections.

\subsection{Closure and Strong Closure}

Many DBMs have the same concretization. This is a problem, because the
abstract environments that we manipulate are therefore not necessarily
the most precise ones, and this can lead to imprecision. Thus,
usually, an implementation of the Octagon abstract domain maintains
the invariant that it only manipulates ``canonical'' forms of DBMs,
such that $B = \alpha(\gamma_{oct}(B))$. Such ``canonical'' DBMs are
always the best possible representative over all the DBMs with the
same concretization.

An important fact is that we can characterize best abstractions using
the values they contain, and that we have algorithms to compute
them. We expose these characterizations, together with these
algorithms. Moreover, we give a weaker closedness condition over DBMs,
that does not ensure canonicity, but that allows better algorithms
without loss of precision.

\subsubsection{Best abstractions for $\gamma_{pot}$}

A first step is to remark that canonical DBMs always have null
diagonal values. Moreover, canonical DBMs should always verify the
triangular inequality. We call such DBMs \emph{closed} DBMs:
\begin{definition}[Closed DBM]
  A \emph{closed} DBM is a DBM $B$ verifying the two following
  properties:
  \begin{itemize}
  \item $\forall v \in \varpm,\; B_{vv} = 0$
  \item $\forall uvw \in \varpm,\; B_{uw} \leq B_{uv} + B_{vw}$
  \end{itemize}
\end{definition}

Closed DBMs are exactly best abstractions
for~$\gamma_{pot}$~\cite[Theorem 3.3.6]{mine2004phd}. Hence, closed
DBMs always have non-empty concretizations. We do not detail here the
algorithm used to detect the emptiness of the concretization of a DBM
and to compute closures: instead, we refer the interested reader to
previous work~\cite{mine2004phd,bagnara2009weakly}.

\begin{example}
The closure $\alpha(\gamma_{pot}(A))$ of the DBM $A$ as defined in
\exref{ex:firstex} contains the following additional finite entries:
\begin{itemize}
  \item $\forall u \in \varpm,\;\alpha(\gamma_{pot}(A))_{uu} = 0$
  \item $\alpha(\gamma_{pot}(A))_{\overline y\, \overline z} = 4$
    (corresponding to the constraint $\rho(z)-\rho(y) \leq 4$).
\end{itemize}
\end{example}

\subsubsection{Best abstractions for $\gamma_{oct}$}

We now refine the notion of closure to canonical forms for
$\gamma_{oct}$. It is easy to see that, for any non-empty set $S$ of
regular environments, $\alpha(S)_{uv} = \alpha(S)_{\overline
  v\,\overline u}$. Thus, canonical DBMs for $\gamma_{oct}$ will
verify the \emph{coherence} property:

\begin{definition}[Coherent DBM]
  \label{def:coherence}
  A DBM $B$ is \emph{coherent} when:
  \begin{equation*}
    \forall uv\in\varpm,\; B_{uv} = B_{\overline v\,\overline u}
  \end{equation*}
\end{definition}

Moreover, matrix elements of the form $B_{u\overline u}$ (for $u \in
\varpm$) impose interval constraints on values of $\rho$. These
interval constraints can be combined to entail constraints on any
difference of values of $\rho$. For this reason, canonical forms for
$\gamma_{oct}$ will verify the following strong closedness property:
\begin{definition}[Strongly closed DBM]
  A DBM $B$ is \emph{strongly closed} when it is closed and coherent
  and:
  \begin{equation*}
    \displaystyle\forall uv\in\varpm,\; B_{uv}\leq
    \frac{B_{u\overline u}+B_{\overline vv}}{2}
  \end{equation*}
\end{definition}

This condition is necessary and sufficient: strong closedness
characterizes canonical DBMs for $\gamma_{oct}$.

\begin{theorem}
  \label{thm:octstrclosed}
  Let $B$ be a DBM. The two following properties are equivalent:
  \begin{enumerate}[(i)]
  \item $B$ is strongly closed
  \item $\gamma_{oct}(B) \neq \emptyset$ and $B = \alpha(\gamma_{oct}(B))$
  \end{enumerate}
\end{theorem}
\begin{proof} See, e.g.,~\cite[Theorems 4.3.2 and 4.3.3]{mine2004phd}.
\end{proof}

Usually~\cite{bagnara2009weakly}, to compute strong closure, one first
ensures that the given matrix is coherent, then computes a closure
(i.e., a canonical representative in the sense of $\gamma_{pot}$),
and, finally, performing a so-called \emph{strengthening
  step}\footnote{This is actually an improvement of the method
  described initially by Miné~\cite{mine2004phd}.}:

\begin{definition}[Strengthening]
  Let $B$ be a DBM. The \emph{strengthening} of $B$, noted
  $\lsharp\mathcal S(B)$ is defined by:
  \begin{equation*}
    \lsharp\mathcal S(B)_{uv} = \min\left\{\frac{B_{u\overline
        u}+B_{\overline vv}}{2}; B_{uv}\right\}
  \end{equation*}
\end{definition}

The following theorem states the correctness of the strong closure
algorithm sketched above, consisting in computing a closure followed
by a strengthening:

\begin{theorem}
  \label{thm:octstrenghtcomplete}
  Let $B$ be a coherent DBM with $\gamma_{oct}(B) \neq
  \emptyset$. Then:
\begin{equation*}
  \alpha(\gamma_{oct}(B)) = \lsharp\mathcal S(\alpha(\gamma_{pot}(B)))
\end{equation*}
In particular, if $B$ is coherent and closed, then $\lsharp\mathcal
S(B)$ is strongly closed.
\end{theorem}
\begin{proof} See, e.g.,~\cite[Theorem 8.2.7]{jourdan2016phd}.
\end{proof}

\begin{example}
  \label{ex:strclos}
  In order to consider the strong closure of the DBM $A$ as defined in
  \exref{ex:firstex}, we first need to make it coherent: let
  $\tilde{A}$ be the DBM containing the same entries as $A$, except
  that $\tilde{A}_{z \overline y} = 1$.

  The closure of $\tilde{A}$ contains the following addional finite
  entries:
  \begin{itemize}
  \item $\forall u \in \varpm,\;\alpha(\gamma_{pot}(\tilde{A}))_{uu} = 0$
  \item $\alpha(\gamma_{pot}(\tilde{A}))_{zy} = \alpha(\gamma_{pot}(A))_{\overline y\, \overline z} = 4$ (corresponding to the constraint $\rho(z)-\rho(y) \leq 4$)
  \item $\alpha(\gamma_{pot}(\tilde{A}))_{z\overline z} = 5$ (corresponding to the constraint $\rho(z) \leq 5/2$).
  \end{itemize}
  The strong closure $\alpha(\gamma_{oct}(\tilde{A}))$ is then
  obtained by strengthening $\alpha(\gamma_{pot}(\tilde{A}))$. The
  strengthening operation creates the following new entries:
  \begin{itemize}
  \item $\alpha(\gamma_{oct}(\tilde{A}))_{xy} = \alpha(\gamma_{oct}(\tilde{A}))_{\overline y\, \overline x} = 2$
  \item $\alpha(\gamma_{oct}(\tilde{A}))_{x\overline z} = \alpha(\gamma_{oct}(\tilde{A}))_{z\overline x} = 3$.
  \end{itemize}
\end{example}

\subsection{Weak Closedness}

Usually, the implementations of the Octagon abstract domain maintain
all DBMs strongly closed, so that maximal information is known when
performing an abstract operation. However, this breaks sparsity:
indeed, matrix elements of the form~$B_{u\overline u}$ are
non-relational interval bounds on the variables: as we expect many
variables to be bounded, the strengthening step gives finite bounds
for many DBM cells, and a strengthened DBM loses most of the
sparsity. In general, a DBM has a quadratic size in the number of
variables, and therefore this loss of sparsity is costly. Previous
attempts at improving performances using
sparsity~\cite{singh2015making} did not make this observation. We
believe that, when using these implementations, DBMs quickly become
dense, hence reducing the efficiency of sparse algorithms.

In our algorithms, we propose to skip the strengthening step: instead
of maintaining the invariant that all the manipulated DBMs are
strongly closed, we maintain the invariant that they are \emph{weakly}
closed:

\begin{definition}[Weakly closed DBM]
  \label{oct:weakclosed}
  Let $B$ be a DBM. We say that $B$ is \emph{weakly closed} when any
  of the two following \emph{equivalent} statements hold:
  \begin{enumerate}[(i)]
  \item $B$ has a null diagonal and $\lsharp\mathcal S(B)$ is strongly closed;
  \item $B$ has a null diagonal, $\lsharp\mathcal S(B)$ is coherent, and:
    \begin{equation}
      \label{eq:octweackclos}
      \forall u v w,\;\lsharp\mathcal S(B)_{uw} \leq B_{uv} + B_{vw}
    \end{equation}
  \end{enumerate}
\end{definition}
\begin{proof} The proof of equivalence of the definitions is in~\cite[Definition 8.2.5]{jourdan2016phd}.
\end{proof}

In order to make sure we do not lose precision, we will prove for each
of those operators that it computes abstract values with the same
concretization as with the usual algorithms. Equivalently, we prove
that the strengthening of the abstract values computed by our
operators are equal to the abstract values computed by the usual
operators on the strengthened parameters.

A weakly closed DBM is neither necessarily strongly closed nor
closed. However, a closed and coherent DBM is always weakly closed:
this helps us easily building weakly closed DBMs from arbitrary sets
of octagonal constraints.

\begin{example}
  Continuing on the definitions of \exref{ex:strclos},
  $\alpha(\gamma_{pot}(\tilde{A}))$ is closed and coherent hence
  weakly closed. This DBM contains no entry relating the variable $x$
  and the other variables. This is an improvement in sparsity compared
  to the strong closure $\alpha(\gamma_{oct}(\tilde{A}))$. To the best
  of our knowledge, this opportunity is not leveraged by previously
  known algorithms, such as~\cite{singh2015making}.
\end{example}

This notion of weak closedness has been introduced by Bagnara et
al.~\cite[Appendix A]{bagnara2009weakly} as an intermediate notion for
proving the correctness of the tight closure algorithm (see
\secref{sec:octTighten}). To the best of our knowledge, the use of
weak closedness as an invariant for manipulating sparse DBMs is an
original result of our work.

\section{Operations on Difference Bound Matrices}
\label{sec:dbmops}

The abstract domain of octagons defines several operations
manipulating difference bound matrices. They include lattice
operations, like comparison and join, and abstract transfer functions,
which model state change in the program.

In this section, we recall the standard definition of these
operations, and give the new sparsity-preserving definition on weakly
closed DBMs. All these algorithms preserve the sparsity and weak
closedness of DBMs and can be proved to be as precise as the standard
ones. More precisely, we claim that they always return DBMs whose
strengthening equals the DBMs that would have been returned by the
traditional algorithms. The implementation of the widening operation,
detailed in~\cite[Section 8.2.7]{jourdan2016phd}, is more complex and
omitted by lack of space.

\subsection{Comparison}
\label{sec:octwcomp}

In order to use octagons in a static analyzer, we need to define a
comparison operator, taking two DBMs and returning a Boolean. If this
Boolean is \verb+true+, then we have the guarantee that the
concretization of the first operand is included in that of the second
operand.

A good candidate is $\dbmleq$, the natural order relation between
DBMs. Its soundness is guaranteed by the monotonicity of
$\gamma_{oct}$. In usual implementations of the Octagon abstract
domain, DBMs are kept strongly closed, hence this operator is actually
as precise as possible: it returns \verb+true+ if and only if the
concretizations are included.

However, in the setting of weakly closed DBMs, this property does not
hold. In order not to lose precision while still using sparse DBM, we
need another comparison operator that strengthens the bounds of the
left operand when they do not entail the right operand:

\begin{definition}[Weakly closed comparison]
  Let $A$ and $B$ be DBMs. The weakly closed comparison of $A$ and
  $B$, noted \mbox{$A \dbmleq_{weak} B$} is defined by:
  \begin{equation*}
    A \dbmleq_{weak} B
    \equiv
    \bigwedge_{\substack{u,v \in \varpm \\ B_{uv} < +\infty}}
    A_{uv} \leq B_{uv} \;\vee\; \frac{A_{u\overline u} + A_{\overline vv}}{2} \leq B_{uv}
  \end{equation*}
\end{definition}

That is, for every finite bound on $B$, we first check whether it is
directly entailed by the corresponding bound in $A$, and then try to
entail it using non-relational bounds. The following theorem states
that it implements the comparison on concretizations, hence we can use
it in a sparse context without losing precision:

\begin{theorem}
  Let $A$ a weakly closed DBM and $B$ any DBM. The two following
  statements are equivalent:
  \begin{enumerate}[(i)]
  \item $\gamma_{oct}(A) \subseteq \gamma_{oct}(B)$
  \item $A \dbmleq_{weak} B$
  \end{enumerate}
\end{theorem}
\begin{proof} See, e.g.,~\cite[Theorem 8.2.9]{jourdan2016phd}.
\end{proof}

\subsection{Forgetting Variables}
\label{sec:octwforg}

An important operation provided by the Octagon abstract domain is
\verb+forget+. When given a DBM and a variable $v$, it returns another
DBM where all the information on $v$ has been forgotten. Its concrete
and abstract definitions are given by:

\begin{definition}[Concrete forgetting]
  \label{def:octConcForget}
  Let $x\in\varplus$ and $S$ be a set of regular environments. We
  define:
  \begin{equation*}
    \mathcal F^x_{oct}(S) = \{ \sigma + [x \Rightarrow r ; \overline x
    \Rightarrow -r] \mid \sigma \in S, r \in \mathbb R \}
  \end{equation*}
\end{definition}

\begin{definition}[Abstract forgetting]
  \label{def:octAbsForget}
  \begin{enumerate}
  \item Let $x\in\varpm$ and $B$ be a DBM. We define $\lsharp\mathcal
    F^x_{pot}(B)$ the DBM such that:
    \begin{equation*}
      \lsharp\mathcal F^x_{pot}(B)_{uv} =
      \begin{cases}
        0 & \text{if $u = v = x$} \\
        +\infty & \text{otherwise if $u = x$ or $v = x$} \\
        B_{uv} & \text{otherwise}
      \end{cases}
    \end{equation*}
  \item Let $x\in\varplus$ and $B$ a DBM. We define:
    \begin{equation*}
      \lsharp\mathcal F^x_{oct}(B) = \lsharp\mathcal F^x_{pot}(\lsharp\mathcal F^{\overline x}_{pot}(B))
    \end{equation*}
  \end{enumerate}
\end{definition}

It is a known result from the Octagon literature~\cite[Theorems 3.6.1
and 4.4.2]{mine2004phd} that $\lsharp\mathcal F^x_{oct}$ is sound when
applied to any DBM. Moreover, when applied to any strongly closed DBM,
it is exact and returns a strongly closed DBM. To these properties, we
add similar properties for weak closedness, that let us use
$\lsharp\mathcal F^x_{oct}$ as-is for weakly closed DBMs without loss
of precision:

\begin{theorem}
  Let $B$ be a weakly closed DBM and $x \in \varplus$. We have:
  \begin{enumerate}
  \item $\lsharp\mathcal S(\lsharp\mathcal F^x_{oct}(B)) = \lsharp\mathcal F^x_{oct}(\lsharp\mathcal S(B))$
  \item $\mathcal F^x_{oct}(\gamma_{oct}(B)) = \gamma_{oct}(\lsharp\mathcal F^x_{oct}(B))$
  \item $\lsharp\mathcal F^x_{oct}(B)$ is weakly closed
  \end{enumerate}
\end{theorem}
\begin{proof} See~\cite[Theorem 8.2.11]{jourdan2016phd}.
\end{proof}

\subsection{Join}
\label{sec:octwjoin}

The usual join operator on DBMs is the least upper bound operator
for~$\dbmleq$:

\newcommand\dbmjoin{\mathbin{\lsharp{\,\cup}}}

\begin{definition}[DBM least upper bound]
  Let $A$ and $B$ be two DBMs. The least upper bound
  $\dbmjoin$ on DBMs is defined by:
  \begin{equation*}
    \forall uv,\;(A\dbmjoin B)_{uv} = \max\{A_{uv}\;;\;B_{uv}\}
  \end{equation*}
\end{definition}

The order relation $\dbmleq$ and the operator $\dbmjoin$ clearly form
an upper semi-lattice, thus usual properties on Galois connections
hold, providing the usual results on the soundness and precision of
this operator: $\dbmjoin$ is sound, and, if given strongly closed
DBMs, it returns the best strongly closed DBM approximating the
concrete union.

For weakly closed DBMs, even though $\dbmjoin$ is sound, it would
possibly lose precision when applied to non-strongly closed DBMs. For
example, the weakly closed DBM $A$ represents the two following
inequalities on positive variables $x$ and $y$:
\begin{align*}
  x+x &\leq 1 & y+y &\leq 0
\end{align*}
The weakly closed DBM $B$, in turn, represents the two following
inequalities:
\begin{align*}
  x+x &\leq 0 & y+y &\leq 1
\end{align*}
The inequality $x+y \leq 1/2$ is not present in $A$ nor in $B$, even
though it is in $\lsharp\mathcal S(A)$ and in $\lsharp\mathcal
S(B)$. As a result, $A \dbmjoin B$ contains the inequalities $x+x \leq
1$ and $y+y \leq 1$, but does not entail $x+y \leq 1/2$, which is
entailed however by $\lsharp\mathcal S(A) \dbmjoin \lsharp\mathcal
S(B)$.

The rationale behind this example is that a join can create some
amount of relationality that was not present in one or both
operands. Our operator has to reflect this fact. Care should be taken,
however, not to break the sparsity of the operands by introducing
spurious finite values in the matrix. Our join for weakly closed DBMs
is defined as follows:

\begin{definition}[Weakly closed join for octagons]
  Let $A$ and $B$ be two weakly closed DBMs. We take, for $u, v \in
  \varpm$, $B^{1/2}_{uv} = \frac{B_{u\overline u}+B_{\overline
      vv}}{2}$ and $A^{1/2}_{uv} = \frac{A_{u\overline u}+A_{\overline
      vv}}{2}$.  The weakly closed join $\dbmjoin_{weak}$ is defined
  in two steps:
  \begin{enumerate}
  \item We first define $A \dbmjoin^0_{weak} B$. Let $u, v \in
    \varpm$. We define:
    \begin{equation*}
      (A \dbmjoin^0_{weak} B)_{uv} =
      \begin{cases}
        A_{uv} & \text{if $A_{uv} = B_{uv}$} \\
        B_{uv} & \text{if $A_{uv} < B_{uv} \leq B^{1/2}_{uv}$} \\
        \max\{ A_{uv}\;;\; B^{1/2}_{uv} \} &
        \text{if $A_{uv} < B_{uv} \wedge B^{1/2}_{uv} < B_{uv}$} \\
        (B \dbmjoin^0_{weak} A)_{uv} & \text{if $A_{uv} > B_{uv}$}
      \end{cases}
    \end{equation*}
  \item Let $u, v \in \varpm$. We define:
    \begin{equation*}
      (A \dbmjoin_{weak} B)_{uv} =
      \begin{cases}
        \min\left\{
        \begin{matrix}
          (A \dbmjoin^0_{weak} B)_{uv} \\
        \max\left\{A^{1/2}_{uv} \;;\; B^{1/2}_{uv}\right\}
        \end{matrix}
        \right\}
        & \begin{matrix}
           \text{if }A_{u\overline u} <
              B_{u\overline u} \wedge A_{\overline vv} >
              B_{\overline vv} \\
            \text{or }A_{u\overline u} >
              B_{u\overline u} \wedge A_{\overline vv} <
              B_{\overline vv}
          \end{matrix} \\
        (A \dbmjoin^0_{weak} B)_{uv} & \text{otherwise}
      \end{cases}
    \end{equation*}
  \end{enumerate}
\end{definition}

The first step can be computed by iterating over all the matrix
elements that are different in $A$ and $B$. This first step thus
preserves the sparsity, and consumes computing time only for variables
that are different in both branches. The second step can be computed
efficiently by first collecting in a list all the variables $u$ for
which $A_{u\overline u} < B_{u\overline u}$ and, in another list, all
those for which $B_{u\overline u} < A_{u\overline u}$. By iterating
over the two lists, we can efficiently modify only the cells meeting
the given condition. It should be noted that we break in the second
step only the sparsity that needs to be broken, as the modified cells
correspond to the cases where the join create new relational
information (as in the example above).

The following theorem states that this modified join operator can be
used on weakly closed DBMs without losing precision or soundness:

\begin{theorem}
\label{thm:octjoincorrect}
Let $A$ and $B$ be two weakly closed DBMs. We have:
\begin{enumerate}
\item $\lsharp\mathcal S(A \dbmjoin_{weak} B) = \lsharp\mathcal S(A) \dbmjoin \lsharp\mathcal S(B)$
\item $\gamma_{oct}(A \dbmjoin_{weak} B) = \gamma_{oct}(\alpha(\gamma_{oct}(A) \cup  \gamma_{oct}(B)))$
\item $A \dbmjoin_{weak} B$ is weakly closed
\end{enumerate}
\end{theorem}
\begin{proof} See~\cite[Theorem 8.2.13]{jourdan2016phd}.
\end{proof}

\subsection{Assuming Constraints}
\label{sec:octwassume}

An important operation for abstract domains is the \verb+assume+
primitive, which refines the internal state of an abstract domain
using a new assumption over the set of approximated environments. In
this section, we only consider the cases where this operation is
exact, i.e., it does not lead to any approximation. These cases amount
to assuming that $\rho(x)-\rho(y) \leq C$, for $C\in\mathbb R$ and $x$
and $y$ two variables. In order to deal with arbitrary linear
inequalities or even arbitrary arithmetical constraints, it is
necessary to write some supporting module for the Octagon domain that
will translate arbitrary constraints into exact ones. Such a support
module is out of the scope of this paper: we refer the reader
to~\cite{mine2004phd} for more detail. Moreover, note that the
combination of \verb+assume+ together with the \verb+forget+ let us
emulate variable assignment\footnote{An efficient implementation would
  however use a specific, optimized implementation for assignments.},
hence we do not detail variable assignment in this paper.

We give the \verb+assume+ primitive in two versions: one adapted to
$\gamma_{pot}$, and one adapted to $\gamma_{oct}$. We first give the
concrete semantics of this operation, which is the same for irregular
and regular environments:

\begin{definition}[Assuming constraints in the concrete]
  Let $C \in \mathbb{R}$, $x,y\in\varpm$ and $S$ be a set of irregular
  environments. We define:
  \begin{equation*}
    \mathcal A^{x-y \leq C}(S) =
    \{ \sigma \in S \mid \sigma(x) - \sigma(y) \leq C \}
  \end{equation*}
\end{definition}

It is easy to see that we can reflect exactly this operation in
DBMs. Indeed, it suffices to change the cell corresponding to the new
constraint, if the old value is larger than the new one. However, this
does not maintain any kind of closedness, whether it be the normal
closure, the strong closure or the weak closedness. As a result, it is
necessary to run a closure algorithm when inserting the new
constraint. These algorithms are costly (i.e., cubic complexity), and
do not leverage the fact that the input matrix is already almost
closed. For this reason, \emph{incremental closure algorithms} have
been developed, with quadratic complexity. We give here a slightly
different presentation of these algorithms as the one originally given
by Miné~\cite{mine2004phd}:

\newcommand\sharpa{\lsharp{\!\mathcal A}}

\begin{definition}[Assuming constraints in the abstract]
  \label{def:octassumedef}
  Let $C \in \mathbb{R}$, $B$ be a DBM and $x,y\in\varpm$.
  \begin{enumerate}
  \item We define $\sharpa^{x-y\leq C}_{pot}(B)$ the DBM such that,
    for $u,v\in\varpm$:
    \begin{equation*}
      \sharpa^{x-y\leq C}_{pot}(B)_{uv} = \min\{ B_{uv}\;;\; B_{ux} +
      C + B_{yv} \}
    \end{equation*}
  \item If $x,y\in\varplus$, we define $\sharpa^{x-y\leq C}_{weak}(B)$
    and $\sharpa^{x-y\leq C}_{oct}(B)$ as:
    \begin{align*}
      \sharpa^{x-y\leq C}_{weak}(B) &= \sharpa^{\overline y-\overline x\leq C}_{pot}(\sharpa^{x-y\leq
        C}_{pot}(B)) \\
      \sharpa^{x-y\leq C}_{oct}(B) &= \lsharp\mathcal
      S(\sharpa^{x-y\leq C}_{weak}(B))
    \end{align*}
  \end{enumerate}
\end{definition}

It is well-known~\cite[Theorem 8.2.14]{jourdan2016phd} that
$\sharpa^{x-y\leq C}_{oct}$ is sound and exact when applied to a DBM
with a null diagonal. When applied to a strongly closed DBM $B$ with
$0 \leq C + B_{yx}$, the result is strongly closed. Therefore, an
implementation of the \texttt{assume} primitive in the strongly closed
setting first checks whether $0 \leq C + B_{yx}$. If so, it returns
$\sharpa^{x-y\leq C}_{oct}$; otherwise it returns $\bot$.

In particular, when applied to weakly closed DBMs, $\sharpa^{x-y\leq
  C}_{oct}$ is sound and exact, since weakly closed DBMs have null
diagonals. However, because this operator uses $\lsharp\mathcal S$, it
breaks sparsity. The advantage of using weakly closed DBMs is that, in
the setting of weakly closed DBMs, $\lsharp\mathcal S$ is no longer
needed: $\sharpa^{x-y\leq C}_{weak}$ can be used as-is, provided the
implementation additionally checks that $0 \leq 2C + B_{y\overline y}
+ B_{\overline xx}$. The following theorem summarizes this result, and
justifies the use of this transfer function in the context of sparse
DBMs without loss of precision:

\begin{theorem}
  Let $C \in \mathbb{R}$, $B$ a weakly closed DBM and $x,y \in
  \varplus$. We have:
  \begin{enumerate}
  \item If $0 \leq 2C + B_{y\overline y} +
      B_{\overline xx}$, then
    $\lsharp\mathcal S(\sharpa^{x-y\leq C}_{weak}(B))
    =
    \sharpa^{x-y\leq C}_{oct}(\lsharp\mathcal S(B))$
  \item $\gamma_{oct}(\sharpa^{x-y\leq C}_{weak}(B)) =
    \mathcal A^{x-y\leq C}(\gamma_{oct}(B))$
  \item If $B$ is weakly closed, the following statements are
    equivalent:
    \begin{enumerate}[(i)]
    \item $\mathcal A^{x-y\leq C}(\gamma_{oct}(B)) \neq \emptyset$
    \item $0 \leq \sharpa^{x-y\leq C}_{weak}(B)_{xx}$
    \item $0 \leq C + B_{yx}$ and $0 \leq 2C + B_{y\overline y} +
      B_{\overline xx}$
    \item $\sharpa^{x-y\leq C}_{weak}(B)$ is weakly closed.
    \end{enumerate}
  \end{enumerate}
\end{theorem}
\begin{proof} See, e.g.,~\cite[Theorem 8.2.15]{jourdan2016phd}.
\end{proof}

\subsection{Tightening}
\label{sec:octTighten}

Miné~\cite{mine2004phd} and Bagnara et al.~\cite{bagnara2009weakly}
study the case of the Octagon abstract domain when the considered
environments take only values in $\mathbb{Z}$: in contrast with the
previous sections, in this case, the strongly closed DBMs are not all
canonical, so that modified algorithms need to be used. We explain
here that the use of the weakly closed setting is compatible with the
integer case. To this end, we define a different concretization
function, $\gamma^{\mathbb{Z}}_{oct}$, that concretizes to integer
environments:

\begin{definition}[Integer concretization of octagons]
  Let $B$ be a DBM. We define:
  \begin{equation*}
    \gamma^{\mathbb{Z}}_{oct}(B) = \{ \rho \in \gamma_{oct}(B) \mid \forall
    u\in\varplus,\;\rho(u) \in \mathbb{Z}\}
  \end{equation*}
\end{definition}

If we consider only integer environments, best abstractions have a
slightly stronger characterization. Such DBM are said \emph{tightly
  closed}. We also define the notion of \emph{weakly tightly closed}
DBMs, which is the analog of \emph{weakly closed} DBMs for the integer
case:

\begin{definition}[Tight closure]
  Let $B$ be a DBM. $B$ is \emph{tightly closed} (respectively
  weakly tightly closed) when:
  \begin{itemize}
  \item $B$ is strongly closed (respectively weakly closed)
  \item $\forall uv \in \varpm,\;B_{uv} \in \mathbb{Z}$
  \item $\forall u \in \varpm,\;\frac{B_{u\overline u}}{2} \in \mathbb{Z}$
  \end{itemize}
\end{definition}

Tightly closed DBMs are exactly best abstractions for integer
environments~\cite[Theorem 8.2.17]{jourdan2016phd}. Bagnara et al.\
\cite[\S 6]{bagnara2009weakly} give efficient algorithms for computing
the tight closure of a DBM. It consists in using a \emph{tightening}
operation before strengthening. The \emph{tightening} operation is
defined by:

\begin{definition}[Tightening]
  Let $B$ a DBM with elements in $\mathbb{Z}$. We define
  $\lsharp\mathcal{T}(B)$ be the DBM with elements in $\mathbb{Z}$
  such that, for $u,v\in\varpm$:
  \begin{equation*}
    \lsharp\mathcal{T}(B)_{uv} =
    \begin{cases}
      B_{uv} - 1 & \text{if $u = \overline v$ and $B_{uv}$ is odd}\\
      B_{uv} & \text{otherwise}
    \end{cases}
  \end{equation*}
\end{definition}

The following theorem gives the essential property of the tightening
operation:

\begin{theorem}
  \label{thm:octtightening}
  Let $B$ a weakly closed DBM with elements in $\mathbb{Z}$. We
  suppose that $\forall u\in\varpm,\;0 \leq
  \lsharp\mathcal{T}(B)_{u\overline u} +
  \lsharp\mathcal{T}(B)_{\overline uu}$. Then $\lsharp\mathcal{T}(B)$
  is weakly tightly closed.
\end{theorem}
\begin{proof} See, e.g.,~\cite[Theorem 8.2.18]{jourdan2016phd}.
\end{proof}

This theorem has two consequences. First, as already explained by
Bagnara et al.~\cite[\S 6]{bagnara2009weakly}, it gives an efficient
algorithm to compute tight closure: one would compute the closure of
the input matrix, then tighten it and finally strengthen it. Second,
our sparse algorithms need only small adjustments when used with
integer environments: instead of maintaining the DBMs weakly closed,
we just have to make them weakly tightly closed by tightening them
after each operation.

Note, however, that tightening does not address the case of mixed
environments, where some variables are known to have integer values,
and some others can have an arbitrary real values. To the best of our
knowledge, there is no known efficient closure algorithm supporting
this use case, even in the dense setting.

\section{Conclusion}
\label{sec:conclusion}

In this paper, we presented new algorithms for the Octagon abstract
domain, which preserve the sparsity of the representation of
octagons. These algorithms are as precise as the usual ones, and rely
on a weaker invariant over difference bound matrices, called weak
closedness. We have shown that these algorithms can be used in the
context of rational or real environments as well as in the context of
integer environments.

We implemented and formally verified in Coq these algorithms in the
context of the Verasco static
analyzer~\cite{jourdan2015formally,jourdan2016phd,verasco}. The use of
these new algorithms improved the performances of the Octagon abstract
domain by at least one order of magnitude.

There are still possible improvements to these algorithms: in
particular, we think that it could be profitable to sparsify
difference bound matrices as much as possible after each abstract
operation, while still maintaining them weakly closed. Indeed,
abstract operations may infer bounds in difference bound matrices that
can actually be deduced from non-relational bounds, therefore missing
opportunity of sparsity.

We think the reduction algorithm presented by Bagnara et
al.~\cite{bagnara2009weakly} can be adapted to compute reduced
difference bound matrices using only weakly closed difference bound
matrices. This would lead to a simpler widening algorithm based on a
semantic definition as described by Bagnara et al.~\cite[\S
4.2]{bagnara2009weakly}. We believe the implementation of these new
algorithms in state-of-the-art static analyzers, by using, for
example, the framework developed by Singh et
al.~\cite{singh2015making} would lead to a significant performance
improvement.

\bibliographystyle{entcs}
\bibliography{bib}

\end{document}